%% file: shallow_packing.tex
\documentclass[letterpaper,11pt,onecolumn]{article}

\usepackage[]{epsf,epsfig, amsmath,amssymb,amsthm, latexsym, color,graphicx, xspace, url}
\graphicspath{{Figures/}}

\usepackage{times,euscript, balance}

\usepackage{xspace}

\newtheorem{theorem}{Theorem}[section]
\newtheorem{lemma}[theorem]{Lemma}

\newtheorem{corollary}[theorem]{Corollary}
\newtheorem{definition}[theorem]{Definition}

\def\eps{{\varepsilon}}
\def\polylog{{\rm polylog}}
\def\A{\EuScript{A}}

\def\D{\EuScript{D}}

\def\F{\EuScript{F}}

\def\M{\EuScript{M}}

\def\P{\EuScript{P}}

\def\S{\EuScript{S}}
\def\T{\EuScript{T}}

\def\V{\EuScript{V}}
\def\W{\EuScript{W}}

\def\etal{\textit{et~al.}}

\def\disc{{\rm{disc}}}
\def\reals{{\mathbb R}}
\def\integers{{\mathbb Z}}

\def\uu{{\bf u}}
\def\vv{{\bf v}}
\def\ww{{\bf w}}

\DeclareMathOperator{\EE}{\mathbf{Exp}}
\DeclareMathOperator{\Var}{\mathbf{Var}}
\DeclareMathOperator{\Prob}{\mathbf{Prob}}

\newcommand{\f}{\mathsf{f}}

\newcommand{\ceil}[1]{\left\lceil {#1} \right\rceil}
\newcommand{\floor}[1]{\left\lfloor {#1} \right\rfloor}

\newdimen\instindent
\def\institute#1{\gdef\@institute{#1}}

 \newfont{\affaddr}{phvr at 11pt}
 \newfont{\affaddrit}{phvro at 11pt} 

\usepackage{graphicx,floatflt,psfrag}

\begin{document}

\title{Shallow Packings in Geometry \thanks{%
    Work on this paper has been supported by NSF under grant CCF-12-16689 and CCF-11-17336. 
  }
}

\author{Esther Ezra
  \thanks{%
    Department of Computer Science and Engineering,
    Polytechnic Institute of NYU, Brooklyn, NY~11201-3840, USA;
    and School of Mathematics,
    Georgia Institute of Technology, Atlanta, Georgia 30332, USA;
    \texttt{esther@courant.nyu.edu}
  }
} 


%
%
%


\maketitle

\begin{abstract}
  We refine the bound on the packing number, originally shown by Haussler, for shallow geometric set systems.
  Specifically, let $\V$ be a finite set system defined over an $n$-point set $X$; we view $\V$ as a set of indicator vectors over the 
  $n$-dimensional unit cube.
  A $\delta$-separated set of $\V$ is a subcollection $\W$, s.t. the Hamming distance between each pair $\uu, \vv \in \W$ is greater than
  $\delta$, where $\delta > 0$ is an integer parameter. The $\delta$-packing number is then defined as the cardinality of the largest 
  $\delta$-separated subcollection of $\V$. Haussler showed an asymptotically tight bound of $\Theta((n/\delta)^d)$ 
  on the $\delta$-packing number if $\V$ has VC-dimension (or \emph{primal shatter dimension}) $d$. 
  We refine this bound for the scenario where, for any subset, $X' \subseteq X$ of size $m \le n$ and for any parameter $1 \le k \le m$, 
  the number of vectors of length at most $k$ in the restriction of $\V$ to $X'$ is only $O(m^{d_1} k^{d-d_1})$, for a fixed integer $d > 0$ 
  and a real parameter $1 \le d_1 \le d$ (this generalizes the standard notion of \emph{bounded primal shatter dimension} when $d_1 = d$).
  In this case when $\V$ is ``$k$-shallow'' (all vector lengths are at most $k$), we show that its $\delta$-packing number 
  is $O(n^{d_1} k^{d-d_1}/\delta^d)$, matching Haussler's bound for the special cases where $d_1=d$ or $k=n$. 
  As an immediate consequence we conclude that set systems of halfspaces, balls, and 
  parallel slabs defined over $n$ points in $d$-space admit better packing numbers when $k$ is smaller than $n$.
  Last but not least, we describe applications to (i) spanning trees of low total crossing number, 
  and (ii) geometric discrepancy, based on previous work by the author.
\end{abstract}

\section{Introduction}
\label{sec:intro}

Let $\V$ be a set system defined over an $n$-point set $X$.
We follow the notation in~\cite{Haussler-95}, and view $\V$ as a set of indicator vectors in 
${\reals}^n$, that is, $\V \subseteq \{0,1\}^n$.
Given a subsequence of indices (coordinates) $I = (i_1, \ldots, i_k)$, $1 \le i_j \le n$, $k \le n$,
the \emph{projection $\V_{|_I}$} of $\V$ onto $I$ (also referred to as the \emph{restriction} of $\V$ to $I$) is defined 
as
$$
\V_{|_I} = \left\{ (\vv_{i_1}, \ldots, \vv_{i_k}) \mid \vv = (\vv_1, \ldots, \vv_n) \in \V \right\} .
$$
With a slight abuse of notation we write $I \subseteq [n]$ to state the fact that $I$ is a subsequence of indices as above.
We now recall the definition of the primal shatter function of $\V$:

\begin{definition}[Primal Shatter Function~\cite{HW87, Mat-99}]
  The \emph{primal shatter function} of $\V \subseteq \{0,1\}^n$ is a function, denoted by $\pi_{\V}$, whose value
  at $m$ is defined by
  $\pi_{\V}(m) = \max_{I \subseteq [n], |I| = m} |\V_{|_I}|$.
  In other words, $\pi_{\V}(m)$ is the maximum possible number of distinct vectors of $\V$
  when projected onto a subsequence of $m$ indices.
\end{definition}

From now on we say that $\V \subseteq \{0, 1\}^n$ 
has \emph{primal shatter dimension $d$} if $\pi_{\V}(m) \le C m^d$, for all $m \le n$,
where $d > 1$ and $C > 0$ are constants.
A notion closely related to the primal shatter dimension is that of the \emph{VC-dimension}:

\begin{definition}[VC-dimension~\cite{Haussler-95, VC-71}]
  An index sequence $I = (i_1,\ldots, i_k)$ is \emph{shattered} by $\V$ if $\V_{|_I} = \{0,1\}^k$.
  The \emph{VC-dimension} of $\V$, denoted by $d_0$ is the size of the longest sequence $I$ shattered by $\V$. 
  That is, $d_0 = \max \{k \mid \exists I = (i_1, i_2, \ldots, i_k), 1 \le i_j \le n, \ \mbox{with}\ \V_{|_I} = \{0,1\}^k \}$.
\end{definition}

The notions of primal shatter dimension and VC-dimension are interrelated.
By the Sauer-Shelah Lemma (see~\cite{Sauer-72, Shelah-72} and the discussion below)
the VC-dimension of a set system $\V$ always bounds its primal shatter dimension, that is, $d \le d_0$. 
On the other hand, when the primal shatter dimension is bounded by $d$, the VC-dimension $d_0$
does not exceed $O(d\log{d})$ (which is straightforward by definition, see, e.g.,~\cite{Har-Peled-11}).

A typical family of set systems that arise in geometry with bounded primal shatter (resp., VC-) dimension consists 
of set systems defined over points in some low-dimensional space ${\reals}^d$, where $\V$ represents a 
collection of certain simply-shaped regions, e.g., halfspaces, balls, or simplices in ${\reals}^d$.
In such cases, the primal shatter (and VC-) dimension is a function of $d$; 
see, e.g.,~\cite{Har-Peled-11} for more details.
When we flip the roles of points and regions, we obtain the so-called \emph{dual set systems} (where 
we refer to the former as \emph{primal set systems}). In this case, the ground set is a collection $\S$
of algebraic surfaces in $\reals^d$, and
$\V$ corresponds to faces of all dimensions in the \emph{arrangement} $\A(\S)$ of $\S$,
that is, 
this is the decomposition of $\reals^d$ into connected open \emph{cells} of dimensions $0, 1, \ldots, d$ induced by $\S$. 
Each cell is a maximal connected region that is contained in the intersection of a fixed number of the surfaces and avoids all other surfaces;
in particular, the $0$-dimensional cells of $\A(S)$ are called ``vertices'', 
and $d$-dimensional cells are simply referred to as ``cells''; see~\cite{SA-95} for more details. 
The distinction between primal and dual set systems in geometry is essential,
and set systems of both kinds appear in numerous geometric applications, see, once again~\cite{Har-Peled-11} and 
the references therein.
 
\paragraph*{$\delta$-packing.}

The \emph{length} $\|\vv\|$ of a vector $\vv \in \V$ under the $L^1$ norm is defined as 
$\sum_{i=1}^n |\vv_i|$, where $\vv_i$ is the $i$th coordinate of $\vv$, $i=1, \ldots, n$.
The \emph{distance} $\rho(\uu,\vv)$ between a pair of vectors $\uu, \vv \in \V$ is
defined as the $L^1$ norm of the difference ${\uu - \vv}$, that is,
$\rho(\uu,\vv) = \sum_{i=1}^n |\uu_i - \vv_i|$.
In other words, it is the \emph{symmetric difference distance} between the corresponding sets represented by $\uu$, $\vv$.

Let $\delta > 0$ be an integer parameter.
We say that a subset of vectors $\W \subseteq \{0,1\}^n$ is \emph{$\delta$-separated} if for each pair $\uu, \vv \in \W$,
$\rho(\uu,\vv) > \delta$.
The \emph{$\delta$-packing number} for $\V$, denote it by $\M(\delta, \V)$, is then defined as the cardinality of the largest 
$\delta$-separated subset $\W \subseteq \V$.
A key property, originally shown by Haussler~\cite{Haussler-95} (see also~\cite{Chazelle-92, CW-89, Dudley-78, Mat-99, Welzl-92}), is that 
set systems of bounded primal shatter dimension admit small $\delta$-packing numbers.
That is:

\begin{theorem}[Packing Lemma~\cite{Haussler-95, Mat-99}]
  \label{thm:packing}
  Let $\V \subseteq \{0,1\}^n$ be a set of indicator vectors of primal shatter dimension $d$, and
  let $1 \le \delta \le n$ be an integer parameter.
  Then
  $\M(\delta, \V) = O((n/\delta)^d)$,
  where the constant of proportionality depends on $d$.
\end{theorem}

We note that in the original formulation in~\cite{Haussler-95} the assumption is that the set system has a finite VC-dimension. 
However, its formulation in~\cite{Mat-99}, which is based on a simplification of the analysis of Haussler by Chazelle~\cite{Chazelle-92},
relies on the assumption that the primal shatter dimension is $d$, which is the actual bound that we state in Theorem~\ref{thm:packing}.
We also comment that a closer inspection of the analysis in~\cite{Haussler-95} shows that this assumption can be replaced with that of 
having bounded primal shatter dimension (independent of the analysis in~\cite{Chazelle-92}). We describe these considerations in 
Section~\ref{sec:prelim}.


\subparagraph*{Previous work.}

In his seminal work, Dudley~\cite{Dudley-78} presented the first application of \emph{chaining}, a proof technique due to Kolmogorov,
to empirical process theory, where he showed the bound $O((n/\delta)^{d_0} \log^{d_0}{(n/\delta)})$ on $\M(\delta, \V)$, 
with a constant of proportionality depending on the VC-dimension $d_0$ (see also previous work by Haussler~\cite{Haussler-92} and 
Pollard~\cite{Pollard-84} for an alternative proof and a specification of the constant of proportionality).
This bound was later improved by Haussler~\cite{Haussler-95}, who showed $\M(\delta, \V) \le e(d_0 + 1)\left(\frac{2e n}{\delta}\right)^{d_0}$ 
(see also Theorem~\ref{thm:packing}), and presented a matching lower bound, which leaves only a constant factor gap, which depends exponentially in $d_0$. 
In fact, the aforementioned bounds are more general, and can also be applied to classes of real-valued functions of finite ``pseudo-dimension'' 
(the special case of set systems corresponds to Boolean functions), see, e.g.,~\cite{Haussler-92}, however, 
we do not discuss this generalization in this paper and focus merely on set systems $\V$ of finite primal shatter (resp., VC-) dimension.

The bound of Haussler~\cite{Haussler-95} (Theorem~\ref{thm:packing}) is in fact a generalization of the so-called 
Sauer-Shelah Lemma~\cite{Sauer-72, Shelah-72}, asserting that 
$|\V| \le (en /d_0)^{d_0}$, where $e$ is the base of the natural logarithm, and thus this bound is $O(n^{d_0})$.
Indeed, when $\delta = 1$, the corresponding $\delta$-separated set should include all vectors in $\V$, 
and then the bound of Haussler~\cite{Haussler-95} becomes $O(n^{d_0})$, matching the Sauer-Shelah bound up to a constant factor
that depends on $d_0$. 

There have been several studies extending Haussler's bound or improving it in some special scenarios.
We name only a few of them.
Gottlieb~\etal~\cite{GKM-12} presented a sharpening of this bound when $\delta$ is relatively large, i.e., $\delta$
is close to $n/2$, in which case the vectors are ``nearly orthogonal''. They also presented a tighter lower bound,
which considerably simplifies the analysis of Bshouty~\etal~\cite{BLL-09}, who achieved the same tightening.

A major application of packing is in obtaining improved bounds on the \emph{sample complexity} in machine learning.
This was studied by Li~\etal~\cite{LLS-01} (see also~\cite{Haussler-92}), who presented an asymptotically tight bound on the sample complexity,
in order to guarantee a small ``relative error.'' This problem has been revisited by Har-Peled and Sharir~\cite{HS-11} in the context of geometric set 
systems, where they referred to a sample of the above kind as a ``relative approximation'' (discussed in Appendix~\ref{app:compact_rep}), 
and showed how to integrate it into an \emph{approximate range counting} machinery, which is a central application in computational geometry.
The packing number has also been used by Welzl~\cite{Welzl-92} in order to construct spanning trees of low 
crossing number (see also~\cite{Mat-99}) and by Matou{\v s}ek~\cite{Mat-95, Mat-99} 
in order to obtain asymptotically tight bounds in geometric discrepancy.
We discuss these applications in the context of the problem studied in this paper in Section~\ref{sec:applications}. 

\paragraph*{Our result.}

In the sequel, we refine the bound in the Packing Lemma (Theorem~\ref{thm:packing}) so that it becomes sensitive to the length 
of the vectors $\vv \in \V$, based on an appropriate refinement of the underlying primal shatter function.
This refinement has several geometric realizations.
Our ultimate goal is to show that when the set system is ``shallow'' (that is, the underlying vectors are short), the packing number
becomes much smaller than the bound in Theorem~\ref{thm:packing}.

Nevertheless, we cannot always enforce such an improvement, as in some settings the worst-case asymptotic bound on the packing number 
is $\Omega((n/\delta)^d)$ even when the set system is shallow.
We demonstrate such a scenario by considering dual set systems of axis-parallel rectangles and points in the plane, where one can have
a large subcollection $\F$ that is both $\delta$-separated and $\delta$-shallow. 
In this case $|\F| = \Omega((n/\delta)^2)$, which is not any better
than the ``standard'' bound (stated in Theorem~\ref{thm:packing}) obtained without the shallowness assumption. 
See Figure~\ref{fig:rectangles} and Appendix~\ref{app:rectangles}, where we give a more detailed description of 
this construction to the non-expert reader.


Therefore, in order to obtain an improvement on the packing number of shallow set systems, 
we may need further assumptions on the primal shatter function. 
Such assumptions stem from  the random sampling technique of Clarkson and Shor~\cite{CS-89},
which we define as follows.
Let $\V$ be our set system.
We assume that for any sequence $I$ of $m \le n$ indices,
and for any parameter $1 \le k \le m$, the number of vectors in $\V_{|_{I}}$
of length at most $k$ is only $O(m^{d_1} k^{d-d_1})$, where $d$ is the primal shatter dimension and
$1 \le d_1 \le d$ is a real parameter.\footnote{We ignore the cases where $d_1 < 1$, 
  as it does not seem to appear in natural set systems---see below.}
When $k = m$ we obtain $O(m^{d})$ vectors in total, in accordance with the assumption that the primal 
shatter dimension is $d$, but the above bound is also sensitive to the length of the vectors as long as $d_1 < d$.
From now on, we say that a primal shatter function of this kind has the \emph{$(d, d_1)$ Clarkson-Shor property}.
 

Let us now denote by $\M(\delta, k, \V)$ the $\delta$-packing number of $\V$, where the vector length of each element in $\V$
is at most $k$, for some integer parameter $1 \le k \le n$. 
By these assumptions, we can assume, without loss of generality, that $k \ge \delta/2$, as otherwise the distance between any two 
elements in $\V$ must be strictly less than $\delta$, in which case the packing is empty.
We also assume that $\delta \le n/2^{(d_0+1)}$
(where $d_0$ is the VC-dim), as otherwise the bound on the packing number is a constant that depends on $d$ and $d_0$ by the Packing Lemma 
(Theorem~\ref{thm:packing}). The choice of this threshold is justified in Section~\ref{sec:analysis} where we present the analysis and show our 
main result, which we state below:   

\begin{theorem}[Shallow Packing Lemma]
  \label{thm:shallow_packing}
  Let $\V \subseteq \{0,1\}^n$ be a set of indicator vectors, whose primal shatter function has a $(d,d_1)$ Clarkson-Shor property, 
  and whose VC-dim is $d_0$.
  Let $\delta$ be an integer parameter between $1$ and $n/2^{(d_0 + 1)}$, $k$ an integer parameter between $1$ and $n$, and suppose that
  $k \ge \delta/2$. Then:
  $$
  \M(\delta, k, \V) = O\left( \frac{ n^{d_1} k^{d-d_1} }{\delta^d} \right) ,
  $$
  where the constant of proportionality depends on $d$.
\end{theorem}

This problem has initially been addressed by the author in~\cite{Ezra-14} as a major tool to obtain size-sensitive discrepancy bounds in set
systems of this kind, where it has been shown $\M(\delta, k, \V) = O\left( \frac{ n^{d_1}  k^{d-d_1} \log^{d}{(n/\delta)}}{\delta^{d}} \right)$.
The analysis in~\cite{Ezra-14} is a refinement over the technique of Dudley~\cite{Dudley-78} combined with the existence of small-size
\emph{relative approximations} (see~\cite{Ezra-14} and Appendix~\ref{app:compact_rep} for more details).
In the current analysis we completely remove the extra $\log^{d}{(n/\delta)}$ factor appearing in the previous bound.
In particular, when $d_1 = d$ (where we just have the original assumption on the primal shatter function) or $k=n$ (in which 
case each vector in $\V$ has an arbitrary length), our bound matches 
the tight bound of Haussler, 
and thus appears as a generalization of the Packing Lemma (when replacing VC-dimension by primal shatter dimension).

Theorem~\ref{thm:shallow_packing} implies smaller packing numbers for several natural geometric set systems
under the shallowness assumption. These set systems are described in detail in Section~\ref{sec:geom_settings}.

Next, in Section~\ref{sec:spanning_trees} we present an application of Theorem~\ref{thm:shallow_packing} to ``spanning trees with low total conflict number'',
which is based on the machinery of Welzl~\cite{Welzl-92} to construct spanning trees of low crossing number (see also~\cite{Mat-99}).
Here the tree spans $\V$ (representing, say, a set of regions defined over $n$ points in $d$-space), and the ``conflict number'' 
of an edge $(u,v)$ is the symmetric difference distance between $u$ and $v$. Based on this structure we introduce a general framework to efficiently 
compute various measures arising in geometric optimization (e.g., diameter, width, radius of the smallest enclosing ball, volume of the minimum bounding box, 
etc.) in each region represented by $\V$, where the key idea is to keep the overall number of updates small (given a spanning tree of the above kind).
In Section~\ref{sec:geom_disc} we show an application in geometric discrepancy, where the goal is to obtain discrepancy bounds 
that are sensitive to the length of the vectors in $\V$. 
Due to the bound in Theorem~\ref{thm:shallow_packing} we obtain an improvement over the one presented in~\cite{Ezra-14}.

In Section~\ref{sec:conclusions} we discuss the geometric interpretation of Theorem~\ref{thm:shallow_packing} to dual set systems.
In particular, we draw the connection between shallow packings and \emph{shallow cuttings}~\cite{Mat-92}. 

Beyond the geometric applications, this paper is primarily an extension of Haussler's technique~\cite{Haussler-95} to shallow set systems.
We note that whereas the analysis of Dudley~\cite{Dudley-78} is fairly simple and intuitive, 
the analysis of Haussler~\cite{Haussler-95} is much more intricate, 
and thus the initial effort in this study was to understand Haussler's analysis, 
whose conclusions are summarized in Appendix~\ref{app:Haussler_proof}. We are also aware of the simplification to Haussler's proof by 
Chazelle~\cite{Chazelle-92}, nevertheless, we had to use the observations made in~\cite{Haussler-95} in order to proceed with
our analysis. Our main conclusion about the analysis in~\cite{Haussler-95} is given in Inequality~(\ref{eq:packing_bound}), 
which implies that the cardinality of a $\delta$-separated set $\V$ is bounded (up to a factor of $(d_0+1)$) by the expected number
of vectors in the projection of $\V$ onto a random sample of $O(d_0 n/\delta)$ indices, where $d_0$ is the VC-dimension. 
Although this simple observation is not explicitly stated in~\cite{Haussler-95}, this relation is a key property used in our
analysis. 

\section{Preliminaries}
\label{sec:prelim}

\paragraph*{Overview of Haussler's Approach.}

For the sake of completeness, we repeat some of the details in the analysis of Haussler~\cite{Haussler-95} 
and use similar notation for ease of presentation.

Let $\V \subseteq \{0,1\}^n$ be a collection of indicator vectors of bounded primal shatter dimension $d$,
and denote its VC-dimension by $d_0$. By the discussion above, $d_0 = O(d\log{d})$.
From now on we assume that $\V$ is $\delta$-separated, and thus a bound on $|\V|$ is also a bound on the packing
number of $\V$.
The analysis in~\cite{Haussler-95} exploits the method of ``conditional variance'' in order to conclude
%
\begin{equation}
  \label{eq:packing_bound}
  |\V| \le (d_0+1)\EE_{I}\left[|\V_{|_{I}}|\right] = O\left(d \log{d} \EE_{I}\left[ |\V_{|_{I}}| \right] \right) ,
\end{equation}
where $\EE_{I}\left[|\V_{|_{I}}|\right]$ is the expected size of $\V$ when projected onto a subset 
$I =\{i_1, \ldots, i_{m-1}\}$ of $m-1$ indices chosen uniformly at random without replacements from $[n]$,
and 
\begin{equation}
  \label{eq:bound_m}
  m := \left\lceil{\frac{(2d_0 + 2)(n+1)}{\delta + 2d_0 + 2} }\right\rceil = O\left(\frac{d_0 n}{\delta}\right) = 
  O\left(\frac{n d\log{d}}{\delta}\right) .
\end{equation}
We justify this choice in Appendix~\ref{app:Haussler_proof}, as well as the facts that $m \le n$ and $I$ consists of precisely $m-1$ indices.

For the sake of completeness, we review Haussler approach in Appendix~\ref{app:Haussler_proof}, and also emphasize some of the properties there,
which are fundamental in our view. Moreover, we refine Haussler's analysis to include two natural extensions:
\emph{(i) Obtain a refined bound on  $\EE_{I}\left[ |\V_{|_I}| \right]$:}
This extension is a direct consequence of Inequality~(\ref{eq:packing_bound}).
In the analysis of Haussler $\EE_{I}\left[ |\V_{|_I}| \right]$ is replaced by its upper bound $O(m^d)$, resulting from the fact that 
the primal shatter dimension of $\V$ (and thus of $\V_{|_I}$) is $d$, from which we obtain that \emph{for any} choice
of $I$, $|\V_{|_{I}}| = O((m-1)^d) = O(m^d)$, with a constant of proportionality that depends on $d$, and thus the packing
number is $O((n/\delta)^d)$, as asserted in Theorem~\ref{thm:packing}.\footnote{We note, however, 
  that the original analysis of Haussler~\cite{Haussler-95} does not rely on the primal shatter dimension, 
  and the bound on $\EE_{I}\left[ |\V_{|_I}| \right]$ is just $O(m^{d_0})$ due to 
  the Sauer-Shelah Lemma. 
}
However, in our analysis we would like to have a more subtle bound on the actual expected value of $|\V_{|_{I}}|$.
In fact, the scenario imposed by our assumptions on the set system eventually yields a much smaller bound on the 
expectation of $|\V_{|_{I}}|$, and thus on $|\V|$. We review this in more detail below.
\emph{(ii) Relaxing the bound on $m$.}
We show that Inequality~(\ref{eq:packing_bound}) is still applicable when the sample $I$ is slightly larger than the 
bound in~(\ref{eq:bound_m}), as a stand alone relation, this may result in a suboptimal bound on $|\V|$, however, 
this property will assist us to obtain local improvements over the bound on $|\V|$, eventually yielding the bound in 
Theorem~\ref{thm:shallow_packing}. 
Specifically, in our analysis, described in Section~\ref{sec:analysis}, we proceed in iterations, where at the first iteration we obtain 
a preliminary bound on $|\V|$ (Corollary~\ref{cor:V_1}), and then, at each subsequent iteration $j > 1$, we draw a sample $I_j$ of $m_j-1$ indices 
where
\begin{equation}
  \label{eq:m_j}
  m_j := m \log^{(j)}{(n/\delta)} = O\left( \frac{d_0 n \log^{(j)}{(n/\delta)}}{\delta} \right), 
\end{equation}
$m$ is our choice in~(\ref{eq:bound_m}), and $\log^{(j)}(\cdot)$ is the $j$th iterated logarithm function.
Then, by a straightforward generalization of Haussler's analysis (described in Appendix~\ref{app:Haussler_proof}),
we obtain, for each $j = 2, \ldots, \log^{*}{(n/\delta)}$:
\begin{equation}
  \label{eq:iter_bound_v}
  |\V| \le (d_0 + 1)\EE_{I_j}\left[ |\V_{|_{I_j}}| \right] .
\end{equation}

We note that since the bounds~(\ref{eq:packing_bound})--(\ref{eq:iter_bound_v}) involve a dependency on the VC-dimension $d_0$, 
we will sometimes need to explicitly refer to this parameter
in addition to the primal shatter dimension $d$. Nevertheless, throughout the analysis we 
exploit the relation $d \le d_0 = O(d\log{d})$, mentioned in Section~\ref{sec:intro}.


\section{The Analysis: Refining Haussler's Approach}
\label{sec:analysis}


\subparagraph*{Overview of the approach.}
We next present the proof of Theorem~\ref{thm:shallow_packing}.
In what follows, we assume that $\V$ is $\delta$-separated.
We first recall the assumption that the primal shatter function of $\V$ has a $(d,d_1)$ Clarkson-Shor property,
and that the length of each vector $\vv \in \V$ under the $L^1$ norm is most $k$.
This implies that $\V$ consists of at most $O(n^{d_1} k^{d-d_1})$ vectors.

Since the Clarkson-Shor property is hereditary, then this also applies to any projection of $\V$ onto a subset of indices, 
implying that the bound on $|\V_{|_{I}}|$ is at most $O(m^{d_1} k^{d-d_1})$, where $I$ is a subset of $m-1$ indices as above. 
However, due to our sampling scheme we expect that the length of each vector in $\V_{|_{I}}$ should be much smaller than $k$, 
(e.g., in expectation this value should not exceed $k (m-1)/n$), from which we may conclude that the actual bound on 
$|\V_{|_{I}}|$ is smaller than the trivial bound $O(m^{d_1} k^{d-d_1})$. 
Ideally, we would like to show that this bound is $O(m^{d_1} (k m/n)^{d-d_1}) = O(n^{d_1}k^{d-d_1}/\delta^d)$, 
which matches our asymptotic bound in Theorem~\ref{thm:shallow_packing} (recall that $m = O(n/\delta)$).
However, this is likely to happen only in case where the length of each vector in $\V_{|_{I}}$ does not exceed its expected value, 
or that there are only a few vectors whose length deviates from its expected value by far, whereas, in the worst case there might be many 
leftover ``long'' vectors in $\V_{|_{I}}$.
Nevertheless, our goal is to show that, with some carefulness one can proceed in iterations, where initially $I$ is a slightly larger sample,
and then at each iteration we reduce its size, until eventually it becomes $O(m)$ and we remain with only a few long vectors. 
At each such iteration $\V_{|_{I}}$ is a random structure that depends on the choice of $I$ 
and may thus contain long vectors, however, in expectation they will be scarce! 

Specifically, we proceed over at most $\log^{*}{(n/\delta)}$ iterations, where we perform local improvements over the bound on $|\V|$, 
as follows.
Let $|\V|^{(j)}$ be the bound on $|\V|$ after the $j$th iteration is completed, $1 \le j \le \log^{*}{(n/\delta)}$.
We first show in 
Corollary~\ref{cor:V_1} that for the first iteration, 
$|\V| \le |\V|^{(1)} = O\left( \frac{n^{d_1} k^{d-d_1} \log^{d}{(n/\delta)}}{\delta^d} \right)$, with a constant of proportionality that 
depends on $d$.
Then, at each further iteration $j \ge 2$, we select a set $I_j$ of 
$m_j - 1 = O(n \log^{(j)}{(n/\delta)}/\delta)$ indices uniformly at random without replacements from $[n]$ 
(see~(\ref{eq:m_j}) for the bound on $m_j$).
Our goal is to bound $\EE_{I_j}\left[ |\V_{|_{I_j}}| \right]$ using the bound $|\V|^{(j-1)}$, obtained at the previous iteration, 
which, we assume by induction to be $O\left( \frac{n^{d_1} k^{d-d_1} (\log^{(j-1)}{(n/\delta)})^{d}}{\delta^d} \right)$ 
(where the base case $j=2$ is shown in Corollary~\ref{cor:V_1}).

A key property in the analysis 
is then to show that the probability that the length of a vector $\vv \in \V_{|_{I_j}}$ (after the projection of $\V$ onto $I_j$) 
deviates from its expectation decays exponentially (Lemma~\ref{lem:exp_decay}).
Note that in our case this expectation is at most $k(m_j - 1)/n$. 
This, in particular, enables us to claim that \emph{in expectation} the overall majority of the vectors in 
$\V_{|_{I_j}}$ have length at most $O(k(m_j-1)/n)$, whereas the remaining longer vectors are scarce.
Specifically, since the Clarkson-Shor property is hereditary, we apply it to $\V_{|_{I_j}}$ and conclude that the number of its vectors 
of length at most $O(k(m_j-1)/n)$ is only $O\left(\frac{n^{d_1} k^{d-d_1} (\log^{(j)}{(n/\delta)})^{d}}{\delta^d} \right)$, 
with a constant of proportionality that depends on $d$. On the other hand, due to Lemma~\ref{lem:exp_decay} and our inductive hypothesis,
the number of longer vectors does not exceed $O\left(\frac{ n^{d_1} k^{d-d_1}}{\delta^d} \right)$, which is dominated by
the first bound. 
We thus conclude $\EE_{I_j}\left[ |\V_{|_{I_j}}| \right] = O\left(\frac{n^{d_1} k^{d-d_1} (\log^{(j)}{(n/\delta)})^{d}}{\delta^d} \right)$.
Then we apply Inequality~(\ref{eq:iter_bound_v}) in order to complete the inductive step, whence we obtain the bound on 
$|\V|^{(j)}$, and thus on $|\V|$. These properties are described more rigorously in Lemma~\ref{lem:bound_projected}, where derive a recursive 
inequality for $|\V|^{(j)}$ using the bound on $\EE_{I_j}\left[ |\V_{|_{I_j}}| \right]$.
We emphasize the fact that the sample $I_j$ is always chosen from the \emph{original} ground set $[n]$, and thus, 
at each iteration we construct a new sample \emph{from scratch}, and then exploit our observation in~(\ref{eq:iter_bound_v}).

\subsection{The First Iteration} 

In order to show our bound on $|\V^{(1)}|$, we form
a subset $I_1 = (i_1, \ldots, i_{m_1})$ of $m_1 = |I_1| = O\left(\frac{d n\log{(n/\delta)}}{\delta} \right)$ indices\footnote{In 
  this particular step we use a different machinery than that of Haussler~\cite{Haussler-95}; see the proof of Lemma~\ref{lem:compact_rep} 
  and our remark after Corollary~\ref{cor:V_1}. Therefore, $|I_1| = m_1$, rather than $m_1-1$. Furthermore,
  the constant of proportionality in the bound on $m_1$ depends just on the primal shatter dimension $d$ instead of the VC-dimension $d_0$
  as in~(\ref{eq:m_j}).
}
with the following two properties: 
(i) each vector in $\V$ is mapped to a distinct vector in $\V_{|_{I_1}}$, and 
(ii) the length of each vector in $\V_{|_{I_1}}$ does not exceed $O(k \cdot m_1/n)$. 

\begin{lemma}
  \label{lem:compact_rep}
  A sample $I_1$ as above satisfies properties (i)--(ii), with probability at least $1/2$.
\end{lemma} 

A set $I_1$ as above exists by the considerations in~\cite{Ezra-14}.
Nevertheless, in Appendix~\ref{app:compact_rep} we present the proof of Lemma~\ref{lem:compact_rep} for the sake 
of completeness and clarity.

We next apply Lemma~\ref{lem:compact_rep} in order to bound $|\V_{|_{I_1}}|$.
We first recall that the $(d,d_1)$ Clarkson-Shor property of the primal shatter function of $\V$ is hereditary.
Incorporating the bound on $m_1$ and property (ii), we conclude that 
$$
|\V_{|_{I_1}}| = O\left(m_1^{d_1} \left(\frac{k m_1}{n}\right)^{d-d_1}\right) =  O\left( \frac{n^{d_1} k^{d-d_1} \log^{d}{(n/\delta)}}{\delta^d} \right) ,
$$
with a constant of proportionality that depends on $d$.
Now, due to property (i), $|\V| \le |\V_{|_{I_1}}|$, we thus conclude:

\begin{corollary}
  \label{cor:V_1}
  After the first iteration we have:
  $|\V| \le |\V|^{(1)} = O\left( \frac{n^{d_1} k^{d-d_1} \log^{d}{(n/\delta)}}{\delta^d} \right)$, 
  with a constant of proportionality that depends on $d$.
\end{corollary}

\noindent{\bf Remark:}
We note that the preliminary bound given in Corollary~\ref{cor:V_1} is crucial for the analysis, as it constitutes the base
for the iterative process described in Section~\ref{sec:iter_process}.
In fact, this step of the analysis alone bypasses our refinement to Haussler's approach, and instead exploits the approach
of Dudley~\cite{Dudley-78}.

\subsection{The Subsequent Iterations: Applying the Inductive Step}
\label{sec:iter_process}



  
Let us now fix an iteration $j \ge 2$. 
As noted above, we assume by induction on $j$ that the bound $|\V|^{(j-1)}$ on $|\V|$ after the $(j-1)$th iteration is
$O\left( \frac{n^{d_1} k^{d-d_1} (\log^{(j-1)}{(n/\delta)})^{d}}{\delta^d} \right)$.
Let $I_j$ be a subset of $m_j-1$ indices, chosen uniformly at random without replacements from $[n]$,
with $m_j$ given by~(\ref{eq:m_j}). Let $\vv \in \V$, and denote by $\vv_{|_{I_j}}$ its projection onto $I_j$.
The expected length $\EE[\| \vv_{|_{I_j}}\|]$ of $\vv_{|_{I_j}}$
is at most $k (m_j-1)/n = O(d_0 k \log^{(j)}{(n/\delta)}/\delta)$.
We next show (see Appendix~\ref{app:exp_decay} for the proof):

\begin{lemma}[Exponential Decay Lemma]
  \label{lem:exp_decay}
  $$
  \Prob\left[\| \vv_{|_{I_j}}\| \ge t \cdot \frac{k (m_j - 1)}{n} \right] <  2^{-t k(m_j - 1)/n},
  $$
  where $t \ge 2e$ is a real parameter and $e$ is the base of the natural logarithm.
\end{lemma}

We now proceed as follows. Recall that we assume $k \ge\delta/2$, and by~(\ref{eq:m_j}) we have
$m_j = O\left(\frac{d_0 n \log^{(j)}{(n/\delta)}}{\delta}\right)$. 
Thus it follows from Lemma~\ref{lem:exp_decay} that
\begin{equation}
  \label{eq:small_prob}
  \Prob\left[\| \vv_{|_{I_j}} \| \ge C \cdot \frac{k (m_j-1)}{n} \right]  <  \frac{1}{( \log^{(j-1)}{(n/\delta)} )^{D} } ,
\end{equation}
where $C \ge 2e$ is a sufficiently large constant, and $D > d_0$ is another constant whose choice depends on $C$ and $d_0$, 
and can be made arbitrarily large. Since $d_0 \ge d$ we obviously have $D > d$.
We next show:

\begin{lemma}
  \label{lem:bound_projected}
  Under the assumption that $k \ge \delta/2$, we have, at any iteration $j \ge 2$:
  \begin{equation}
    \label{eq:recursion}
    |\V|^{(j)} \le A (d_0 + 1)\cdot \frac{n^{d_1} k^{d-d_1} (\log^{(j)}{(n/\delta)})^{d}}{\delta^d} + 
    (d_0 + 1) \cdot \frac{|\V|^{(j-1)}}{(\log^{(j-1)}{(n/\delta)})^{D}} ,
  \end{equation}
  where $|\V|^{(l)}$ is the bound on $|\V|$ after the $l$th iteration,
  and $A > 0$ is a constant that depends on $d$ (and $d_0$) and the constant of proportionality determined by the 
  Clarkson-Shor property of $\V$.
\end{lemma}
\begin{proof}
  We in fact show:
  $$
  \EE_{I_j}\left[ |\V_{|_{I_j}}| \right] \le  A \cdot \frac{n^{d_1} k^{d-d_1} (\log^{(j)}{(n/\delta)})^{d}}{\delta^d} + 
  \frac{|\V|^{(j-1)}}{(\log^{(j-1)}{(n/\delta)})^{D}} ,
  $$
  and then exploit the relation $|\V| \le (d_0 + 1)\EE_{I_j}\left[ |\V_{|_{I_j}}| \right]$ (Inequality~(\ref{eq:iter_bound_v})), 
  in order to prove~(\ref{eq:recursion}). 

  In order to obtain the first term in the bound on $\EE_{I_j}\left[ |\V_{|_{I_j}}| \right]$,
  we consider all vectors of length at most $C \cdot \frac{k (m_j - 1)}{n}$ 
  (where $C \ge 2e$ is a sufficiently large constant as above) 
  in the projection of $\V$ onto a subset $I_j$ of $m_j - 1$ indices (in this part of the analysis $I_j$ can be arbitrary). 
  Since the primal shatter function of $\V$ has a $(d,d_1)$ Clarkson-Shor property, which is hereditary, we obtain at most 
  $$
  O({m_j}^{d_1}  (k(m_j - 1)/n)^{d-d_1}) 
  = O\left(\frac{n^{d_1} k^{d-d_1} (\log^{(j)}{(n/\delta)})^{d}}{\delta^d} \right)
  $$ 
  vectors in $\V_{|_{I_j}}$ of length smaller than $C \cdot \frac{k (m_j - 1)}{n} = O(\frac{k \log^{(j)}{(n/\delta)}}{\delta})$.
  It is easy to verify that the constant of proportionality $A$ in the bound just obtained depends on $d$, $d_0$, and 
  the constant of proportionality determined by the Clarkson-Shor property of $\V$.
  
  Next, in order to obtain the second term, we consider the vectors $\vv \in \V$ that are mapped to vectors $\vv_{|_{I_j}} \in \V_{|_{I_j}}$ with 
  $\| \vv_{|_{I_j}} \| > C \cdot \frac{k (m_j - 1)}{n}$.
  By Inequality~(\ref{eq:small_prob}):
  $$
  \EE\left[\left| \left \{ \vv \in \V  \ \mid \  \| \vv_{|_{I_j}} \| > C \cdot \frac{k (m_j-1)}{n} \right \} \right| \right] <
  \frac{| \V|}{(\log^{(j-1)}{(n/\delta)})^{D}} ,
  $$
  and recall that $|\V|^{(j-1)}$ is the bound on $|\V|$ after the previous iteration $j-1$.
  This completes the proof of the lemma.
  
%
\end{proof}

\noindent{\bf Remark:}
We note that the bound on $\EE_{I_j}\left[ |\V_{|_{I_j}}| \right]$ consists of the \emph{worst-case} bound on the number of short
vectors of length at most $C \cdot k (m_j - 1)/n$, obtained by the Clarkson-Shor property, plus the \emph{expected} number
of long vectors. 

\subparagraph*{Wrapping up.}


We now complete the analysis and solve Inequality~(\ref{eq:recursion}).
Our initial assumption that $\delta \le n/2^{(d_0+1)}$, and the fact that $D > d$ is sufficiently large, 
imply that the coefficient of the recursive term is smaller than $1$, for any $2 \le j \le 1 + \log^{*}{(n/\delta)} - \log^*{(d_0+1)}$.\footnote{We 
  observe that $2 \le 1 + \log^{*}{(n/\delta)} - \log^*{(d_0+1)} \le \log^{*}{(n/\delta)}$, due to our assumption that
  $\delta \le n/2^{(d_0+1)}$, and the fact that $d_0 \ge 1$.} 
Then, using induction on $j$, one can verify that the solution is
\begin{equation}
  \label{eq:bound_V}
  |\V|^{(j)} \le 2A (d_0+1) \frac{n^{d_1} k^{d-d_1} (\log^{(j)}{(n/\delta)})^{d}}{\delta^d} ,
\end{equation}
for any $2 \le j \le 1 + \log^{*}{(n/\delta)} - \log^*{(d_0+1)}$.

We thus conclude $|\V|^{(j)} = O\left(\frac{n^{d_1} k^{d-d_1} (\log^{(j)}{(n/\delta)})^{d}}{\delta^d}\right)$.
In particular, at the termination of the last iteration $j^{*} = 1 + \log^{*}{(n/\delta)} -  \log^*{(d_0+1)}$, we obtain:
$$
|\V| \le |\V|^{(j^{*})} = O\left( \frac{n^{d_1} k^{d-d_1}}{\delta^d} \right) ,
$$
with a constant of proportionality that depends on $d$ (and $d_0$). 
This at last completes the proof of Theorem~\ref{thm:shallow_packing}.

\section{Applications}
\label{sec:applications}

\subsection{Realization to Geometric Set Systems}
\label{sec:geom_settings}

We now incorporate the random sampling technique of Clarkson and Shor~\cite{CS-89} with Theorem~\ref{thm:shallow_packing}
in order to conclude that small shallow packings exist in several useful scenarios. 
This includes the case where $\V$ represents:
(i) A collection of halfspaces defined over an $n$-point set in $d$-space.
In this case, for any integer parameter $0 \le k \le n$, the number of halfspaces that contain at most $k$
points is $O(n^{\floor{d/2}} k^{\ceil{d/2}})$, and thus the primal shatter function has a $(d, \floor{d/2})$ Clarkson-Shor property.
(ii) A collection of balls defined over an $n$-point set in $d$-space.
Here, the number of balls that contain at most $k$ points is $O(n^{\floor{(d+1)/2}} k^{\ceil{(d+1)/2}})$, and therefore
the primal shatter function has a $(d+1, \floor{(d+1)/2})$ Clarkson-Shor property.
(iii) A collection of \emph{parallel slabs} (that is, each of these regions is enclosed between two parallel hyperplanes and has an arbitrary width), 
defined over an $n$-point set in $d$-space.
The number of slabs, which contains at most $k$ points is $O(n^d k)$.
(iv) A dual set system of points in $d$-space and a collection $F$ of $n$ $(d-1)$-variate (not necessarily continuous or totally defined) 
functions of \emph{constant description complexity}. Specifically, the graph of each function is a semi-algebraic set in ${\reals}^d$ defined by 
a constant number of polynomial equalities and inequalities of constant maximum degree (see~\cite[Chapter 7]{SA-95} for a detailed description 
of these properties, which we omit here).\footnote{In~\cite{SA-95} it is also required that the projection of each function onto the plane $x_d = 0$ 
  has a constant description complexity.}
In this case, $\V$ is represented by the cells (of all dimensions) in the \emph{arrangement} of the graphs of the functions in $F$ 
(see Section~\ref{sec:intro} for the definition) that lie below at most $k$ function graphs.
This portion of the arrangement is also referred to as the \emph{at most $k$-level}, and its combinatorial complexity 
is $O(n^{d-1 + \eps} k^{1-\eps})$, for any $\eps > 0$, where the constant of proportionality depends on $d$ and $\eps$.
Thus the primal shatter function has a $(d, d-1+\eps)$ Clarkson-Shor property.


All bounds presented in (i)--(iv) are well known in the field of computational geometry; we refer the reader 
to~\cite{CS-89, Mat-02, SA-95} for further details. We thus conclude:

\begin{corollary}
  \label{cor:halfspaces}
  Let $\V \subseteq \{0,1\}^n$ be a set of indicator vectors representing a set system of halfspaces defined over an $n$-point
  set in $d$-space, and let $\delta, k$ be two integer parameters as in Theorem~\ref{thm:shallow_packing}.
  Then:
  $$
  \M(\delta, k, \V) = O\left( \frac{ n^{\floor{d/2}} k^{\ceil{d/2}}}{\delta^d} \right) ,
  $$
  where the constant of proportionality depends on $d$.
\end{corollary}

\begin{corollary}
  \label{cor:balls}
  Let $\V \subseteq \{0,1\}^n$ be a set of indicator vectors representing a set system of balls defined over an $n$-point
  set in $d$-space, and let $\delta, k$ be two integer parameters as in Theorem~\ref{thm:shallow_packing}.
  Then:
  $$
  \M(\delta, k, \V) = O\left( \frac{ n^{\floor{(d+1)/2}} k^{\ceil{(d+1)/2}}}{\delta^{d+1}} \right) ,
  $$
  where the constant of proportionality depends on $d$.
\end{corollary}

\begin{corollary}
  \label{cor:slabs}
  Let $\V \subseteq \{0,1\}^n$ be a set of indicator vectors representing a set system of parallel slabs defined over an $n$-point
  set in $d$-space, and let $\delta, k$ be two integer parameters as in Theorem~\ref{thm:shallow_packing}.
  Then:
  $$
  \M(\delta, k, \V) = O\left( \frac{n^d k}{\delta^{d+1}} \right) ,
  $$
  where the constant of proportionality depends on $d$.
\end{corollary}

\begin{corollary}
  \label{cor:dual}
  Let $\V \subseteq \{0,1\}^n$ be a set of indicator vectors representing a dual set system of 
  $n$ $(d-1)$-variate (not necessarily continuous or totally defined) functions of constant description complexity
  and points in $d$-space. Let $\delta, k$ be two integer parameters as in Theorem~\ref{thm:shallow_packing}.
  Then:
  $$
  \M(\delta, k, \V) = O\left( \frac{ n^{d-1 + \eps} k^{1-\eps}}{\delta^d} \right) ,
  $$
  where the constant of proportionality depends on $d$ and on $\eps$.
\end{corollary}

\subsection{Spanning Trees of Low Total Conflict Number}
\label{sec:spanning_trees}

Suppose we are given a set $X$ of $n$ points in $d$-space and a set $\Sigma$ of $m$ regions defined over $X$.
With a slight abuse of notation, we also refer to $\Sigma$ as the corresponding set system defined over $X$. 
\footnote{Here, $\Sigma$ has the role $\V$ in our original notation.}
We now assume that the set system $(X, \Sigma)$ is shallow and has a $(d, d_1)$ Clarkson-Shor property. 
The question at hand is to construct a spanning tree over $\Sigma$, whose overall \emph{conflict number} is small.
This notion is related to \emph{spanning trees with low crossing number}, defined by Welzl~\cite{Welzl-92} (see also~\cite{CW-89}), 
where the tree spans $X$ (rather than $\Sigma$), in which case it has a geometric realization,
where the edges of the tree are the line segments connecting points of $X$.
Our structure is dual to that of Welzl~\cite{Welzl-92} and defined as follows.


Let $G = (\Sigma, E)$ be a graph with vertex set $\Sigma$. 
We say that a point $x \in X$ \emph{conflicts} with an edge $\{S, S'\} \in E$ if $x \in S \triangle S'$.
We then define the \emph{conflict number} of an edge $e$ of $G$ as the number of points $x \in X$ with which it is in conflict, 
and then the total conflict number of $G$ is the sum of the conflict numbers, over all its edges.
Using similar arguments as in~\cite[Lemma 5.18]{Mat-99} and~\cite{Welzl-92}, one can show (we omit the easy proof):

\begin{lemma}
  \label{lem:tree}
  Let $\Sigma$ be a set system of $m$ sets, defined over an $n$-point set $X$.
  Assume $\Sigma$ has a $(d, d_1)$ Clarkson-Shor property, and that $|S| \le k$, for each $S \in \Sigma$,
  where $1 \le k \le n$ is an integer parameter.
  Then there exists a spanning tree $\T$ over $\Sigma$, whose total conflict number is $O\left(n^{d_1/d} k^{1-d_1/d} m^{1-1/d}\right)$.
\end{lemma}

\subparagraph*{Constructing an approximation of the tree.}
In order to construct $\T$ efficiently, we relax this problem to only \emph{approximating} the spanning tree,
and then use the machinery of Har-Peled and Indyk~\cite{HI-02} in order to conclude that when $\Sigma$ is a set
system as in Corollaries~\ref{cor:halfspaces}--\ref{cor:slabs}, a $(1+\eta)$-factor approximation for $\T$ (that is,
a spanning tree whose total conflict number is at most $(1+\eta)$ of the smallest such number) can be 
computed in subquadratic time, for any $\eta > 0$. 
We present a sketch of this construction in Appendix~\ref{app:approx_mst}, and conclude:

\begin{corollary}
  \label{cor:approx_mst}
  Let $X$ be a set of $n$ points in $d$-space and let $\Sigma$ be a set of $m$ regions defined over $X$.
  Then, for any $\eta > 0$, one can compute in time 
  
  \noindent
  (i)~$O^{*}\left(n^{\frac{d}{d+1}} m^{\frac{d}{d+1}} + n + m^{1 + \frac{1}{1+\eta}} \right)$ a spanning tree of $\Sigma$, 
  where $\Sigma$ is a set of halfspaces in $d$-space,
  with overall $O\left((1+\eta)n^{\frac{\floor{d/2}}{d}} k^{\frac{\ceil{d/2}}{d}} m^{1-\frac{1}{d}}\right)$ conflicts.
  
  
  \noindent
  (ii)~$O^{*}\left(n^{\frac{d+1}{d+2}} m^{\frac{d+1}{d+2}} + n + m^{1 + \frac{1}{1+\eta}} \right)$, a spanning tree as above, where $\Sigma$ is a set of balls in $d$-space,
  with overall $O\left((1+\eta)n^{\frac{\floor{(d+1)/2}}{d+1}} k^{\frac{\ceil{(d+1)/2}}{d+1}} m^{1-\frac{1}{d+1}}\right)$ conflicts. 
  
  \noindent
  (iii)~$O^{*}\left(n^{\frac{d}{d+1}} m^{\frac{d}{d+1}} + n + m^{1 + \frac{1}{1+\eta}}\right)$, a spanning tree as above, where $\Sigma$ is a set of parallel 
  slabs in $d$-space, with overall $O\left((1+\eta)n^{1 - 1/(d+1)} k^{1/(d+1)} m^{1-\frac{1}{d+1}}\right)$ conflicts.
  
  In the above bounds $O^{*}(\cdot)$ hides a poly-logarithmic factor.
\end{corollary}


Based on the above machinery, we next propose a general framework for updating the optimal solution (or an approximate solution)
of a prescribed geometric optimization problem, over all regions in $\Sigma$. 

\subparagraph*{A general framework.}
Let $X$, $\Sigma$ be as in Lemma~\ref{lem:tree}, and let $\f: \Sigma \rightarrow {\reals}$ be a function that assigns
real values on the sets $S \in \Sigma$ (each being a subset of $X$). For example, $\f$ may correspond to 
diameter, width, radius of the smallest enclosing ball, volume of the minimum bounding box, etc.,
see, e.g.,~\cite{AHV-04}, where these measures are referred to as ``faithful measures''. 
Our goal is to efficiently compute $\f(S)$ for each $S \in \Sigma$. 

Specifically, we assume to have a data-structure $\D$ that maintains a subset $S \subseteq X$ with the following properties:
(i) The time to preprocess $\D$ is $P(n)$.
(ii) The time to update $\D$ (that is, inserting or deleting an element from $\D$) is $U(n)$.
(iii) At any given time, querying $\D$ for the value of $\f$, w.r.t. the set $K$ of the currently stored points, costs $Q(|K|)$ time. 

Having this machinery at hand, in a brute-force approach, $\D$ is initially empty. Then, for each $S \in \Sigma$, we insert its elements into $\D$, 
obtain $\f(S)$ by querying $\D$, and then remove all these elements from $\D$. We proceed in this manner until all sets $S \in \Sigma$ are exhausted.
Under the assumption that $\Sigma$ is $k$-shallow, the resulting running time is $O(P(n) + m k U(n) + m Q(k))$. 
On the other hand, with the existence of a $(1+\eta)$-factor approximation for the spanning tree $\T$ (with properties as in Lemma~\ref{lem:tree}), 
we can proceed as follows. With a slight abuse of notation, we also denote the approximate tree by $\T$.
Initially, $\D$ is empty as above, and we choose an arbitrary set $S \in \Sigma$, for which we compute $\f(S)$ as above.
Then we traverse $\T$ from $S$ in a BFS manner, update $\D$ accordingly, and make a query at each vertex tracked during the search.
Clearly, the number of these updates is proportional to the overall number of conflicts in $\T$, and thus the overall running time is
$O\left(C(n,m,\eta) + P(n) + k U(n) + (1+\eta) n^{\frac{d_1}{d}} k^{1 - \frac{d_1}{d}} m^{1 - \frac{1}{d}} U(n) + m Q(k) \right)$, 
where $C(n,m,\eta)$ is the time to construct (a $(1+\eta)$-factor approximation for) $\T$.
%
We are interested in the scenario where this solution outperforms the brute-force algorithm (at least for some values of $k$).
Below we describe a concrete scenario, related to the approximation of the faithful measures listed above, which demonstrates 
the usefulness of our framework.

\subparagraph*{Dynamic coresets.}
Based on the seminal work of Agarwal~\etal.~\cite{AHV-04} on \emph{coresets}, Chan~\cite{Chan-08} presented a data structure, which
maintains a constant-size coreset, with respect to ``extent'', in $U(n) = O(\log{n})$ update time, for all constant dimensions, with linear 
space and preprocessing time, where the constant of proportionality depends on the error parameter $\eps > 0$.
Using this machinery, it is straightforward to obtain dynamic $(1+\eps)$-factor approximation algorithms with logarithmic update time
for computing the faithful measures stated above. The time $Q(k)$ to compute an approximation for these measures depends on $\eps$
and the dimension $d$. For simplicity of presentation, we omit the dependency on $\eps$ in the bounds of $U(n)$, $Q(k)$,
and compare the performance of our approach w.r.t. the brute force computation, when we consider only the parameters $n$, $m$, and $k$.

With this machinery, the brute-force algorithm runs in $O(m k \log{n} + n)$ time, whereas our algorithm runs in 
$O\left(C(m,n,\eta) + k \log{n} + (1+\eta)n^{\frac{d_1}{d}} k^{1 - \frac{d_1}{d}} m^{1 - \frac{1}{d}} \log{n} + m + n\right)$ time. 
We now consider the bounds stated in Corollary~\ref{cor:approx_mst} and set $\eta = \log{m}$, in which case the term 
$O^{*}(m^{1 + \frac{1}{1+\eta}})$ becomes nearly-linear and we pay only an extra logarithmic factor in the total number of conflicts.
We then conclude:


\begin{corollary}
  \label{cor:compute_measure}
  Let $X$ be a set of $n$ points in $d$-space and a let $\Sigma$ be a set of $m$ regions defined over $X$,
  where, for each $S \in \Sigma$, $|S| \le k$, where $1 \le k \le n$ is an integer parameter.
  Then one can compute a $(1+\eps)$-factor approximation for the aforementioned faithful measures in time that is the minimum
  of $O(m k \log{n} + n)$ and
  
  \noindent
  (i)~$O^{*}\left(n^{\frac{d}{d+1}} m^{\frac{d}{d+1}} + k \log{n} + n^{\frac{\floor{d/2}}{d}} k^{\frac{\ceil{d/2}}{d}} m^{1 - \frac{1}{d}} + m +n\right)$, 
  if $\Sigma$ is a collection of halfspaces, in $d$-space.
  
  
  \noindent
  (ii)~$O^{*}\left(n^{\frac{d+1}{d+2}} m^{\frac{d+1}{d+2}} + k \log{n} + 
  n^{\frac{\floor{(d+1)/2}}{d+1}} k^{\frac{\ceil{(d+1)/2}}{d+1}} m^{1- \frac{1}{d+1}} + m + n\right)$, if $\Sigma$ is a collection of balls in $d$-space.
  
  \noindent
  (iii)~$O^{*}\left(n^{\frac{d}{d+1}} m^{\frac{d}{d+1}} + k \log{n} + 
  n^{\frac{d}{d+1}} k^{\frac{1}{d+1}} m^{1- \frac{1}{d+1}} + m + n\right)$, 
  if $\Sigma$ is a collection of parallel slabs in $d$-space.
  
  The constant of proportionality in each of these bounds depends on $\eps$ and $d$. 
\end{corollary}

\vspace{-2ex}
\subsection{Geometric Discrepancy}
\label{sec:geom_disc}
Given a set system $(X,\Sigma)$ as above, we now wish to color the points of $X$ by two colors, such that in each set of $\Sigma$
the deviation from an even split is as small as possible. 

Formally, a \emph{two-coloring} of $X$ is a mapping $\chi : X \rightarrow \{-1 ,+1\}$.
For a set $S \in \Sigma$ we define $\chi(S) := \sum_{x \in S} \chi(x)$.
The \emph{discrepancy} of $\Sigma$ is then defined as $\disc(\Sigma) := \min_{\chi} \max_{S \in \Sigma} | \chi(S)|$.

In a previous work~\cite{Ezra-14}, the author presented size-sensitive discrepancy bounds for set systems of halfspaces defined over
$n$ points in $d$-space.  
These bounds were achieved by combining the \emph{entropy method}~\cite{LM-12} with $\delta$-packings, and, as observed in~\cite{Ezra-14}, 
they are optimal up to a poly-logarithmic factor.
Incorporating our bound in Theorem~\ref{thm:shallow_packing} into the analysis in~\cite{Ezra-14}, the bounds on $\chi(S)$
improve by a $\sqrt{\log{n}}$ factor. Specifically, we obtain
(we omit the technical details in this version):

\begin{corollary}
  \label{cor:disc_halfspaces}
  Let $\Sigma$ be a set system of halfspaces defined over $n$-points in $d$-space ($d \ge 3$).
  Then, there is a two-coloring $\chi$, such that for each $S \in \Sigma$,
  $\chi(S) = O\left(|S|^{1/4} n^{1/4 - 1/(2d)} \log^{1/2d}{n} \right)$, for $d \ge 4$ even,
  and $\chi(S) = O\left( |S|^{1/4 + 1/(4d)} n^{1/4 - 3/(4d)} \log^{1/2d}{n} \right)$, for $d \ge 5$ odd, 
  where the constant of proportionality depends on $d$.
  When $d = 3$ the bound on $\chi(S)$ is 
  $O\left( |S|^{1/3} \log^{7/6}{n} \right)$.\footnote{We note that the case $d=2$ in Corollary~\ref{cor:disc_halfspaces} 
    has been resolved by Har-Peled and Sharir~\cite{HS-11}.}
\end{corollary}

%


\section{Concluding Remarks and Further Research}
\label{sec:conclusions}

We note that Corollary~\ref{cor:dual} implies that one can pack the ``shallow level'' in 
an arrangement of $(d-1)$-variate function graphs $F$ (as defined in Section~\ref{sec:geom_settings}) with a relatively small number of 
Hamming balls. In fact, if those functions are just hyperplanes, then by a standard point-hyperplane
duality, Corollary~\ref{cor:halfspaces} implies that the number of Hamming balls that pack the at most 
$k$-level in the underlying arrangement is $O(n^{\floor{d/2}} k^{\ceil{d/2}}/\delta^d)$. Roughly speaking, this particular
bound can be obtained by either (i) \emph{shallow cuttings}~\cite{Mat-92} and a standard reduction between packing
and covering. Roughly speaking, this is a coverage of the shallow level of the arrangement by a collection $\Xi$ of (possibly unbounded) 
pairwise-disjoint ``primitive'' cells, such that the interior of each cell is crossed by a small fraction of the function graphs in $F$. 
%
When $F$ is a collection of hyperplanes the packing number and the bound on $|\Xi|$ are asymptotically the same (where each cell of $\Xi$
corresponds to a ``center'' of a Hamming ball). 
Or (ii) the bound $\Omega(\delta^d)$ on the ``volume'' of a Hamming ball of radius $\delta$ (as observed in~\cite{Welzl-92})
and the fact that the complexity of the at most $k$-level in an arrangement of $n$ hyperplanes is 
$O(n^{\floor{d/2}} k^{\ceil{d/2}})$~\cite{CS-89}.\footnote{We omit the official definition of shallow cuttings and the straightforward details in this discussion, 
  and refer the non-expert reader to~\cite{AES-00, Mat-92, Mat-02, Welzl-92} and the references therein.} 
Nevertheless, these approaches may not be applicable when the input functions are more general (e.g., satisfy the properties
stated in Section~\ref{sec:geom_settings} for, say, $d \ge 4$), in this case one may use \emph{linearization} but that may result in an overestimated bound.
We hope our bound on the packing number will be useful in geometric computing, and, in particular
in the context of spanning trees of low (total) crossing number.

The analogy between shallow packings for dual set systems and shallow cuttings may cause one to interpret shallow packings
as the ``primal version'' of shallow cuttings. In this paper we named two useful applications for shallow packings.
We hope to find additional applications in geometry and beyond. 


\subparagraph*{Acknowledgments.}
The authors wishes to thank Boris Aronov, Sariel Har-Peled, Aryeh Kontorovich, and Wolfgang Mulzer for useful discussions and suggestions.
In particular, the author is grateful to Sariel Har-Peled for suggesting the application described in Section~\ref{sec:spanning_trees}.
Last but not least, the author deeply thanks Ramon Van Handel, for various discussions and for spotting an error in an earlier version of this paper.
In particular, work on this paper began due to a discussion with Ramon after a talk the author gave in Princeton University.






\appendix

%

\section{An Example of a Shallow Set System with Large Packing Numbers}
\label{app:rectangles}

\begin{figure*} 
  \centering
  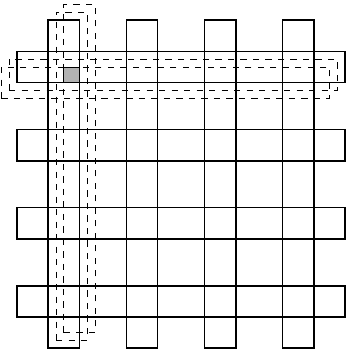  
  \caption{\sf\small{A grid of $\frac{n}{\delta} \times \frac{n}{\delta}$ axis-parallel rectangles, 
    each of which with multiplicity $\delta/2$. The multiplicity in the figure is depicted 
    only for the leftmost vertical and the top horizontal rectangles.
    The small shaded rectangle in the figure is a $\delta$-shallow cell of the arrangement.}}
  \label{fig:rectangles}
\end{figure*}

Continuing our construction from Section~\ref{sec:intro}, 
the ground set is a collection of axis-parallel rectangles, and the vectors $\V$ represent subsets of rectangles 
that cover a point in the plane. For simplicity of exposition, we define these vectors to represent all two-dimensional cells 
in the arrangement of the given rectangles. 
It is well known that the primal shatter function of $\V$ is quadratic 
(see, e.g.,~\cite{Har-Peled-11}), and therefore, by Theorem~\ref{thm:packing}, the packing number is $O((n/\delta)^2)$.
Nevertheless, we claim that even when the arrangement is shallow, the asymptotic bound on the packing number is not any
better than $(n/\delta)^2$. Indeed, fix a positive even parameter $\delta > 0$, and suppose, without loss of generality,
that $n/\delta$ is an integer number. Consider now an $\frac{n}{\delta} \times \frac{n}{\delta}$ 
grid of long and skinny rectangles, where each rectangle in the grid is duplicated $\delta/2$ times 
(with a possibly infinitesimal perturbation), as illustrated in Figure~\ref{fig:rectangles}. 
Clearly, each (two-dimensional) cell 
in the arrangement is covered by at most $\delta$ rectangles, and thus the set system is $\delta$-shallow. 
Consider now only the cells 
at ``depth'' $\delta$ (that is, they are covered by precisely $\delta$ rectangles), and let $\F \subset \V$ 
be the set of their representative vectors.  
It is easy to verify that, for each pair $v, v' \in \F$, the distance $\rho(v,v')$ is at least $\delta$ 
(see once again Figure~\ref{fig:rectangles}), and thus $\F$ is both $\delta$-separated and $\delta$-shallow.
However, by construction, we have $|\F| = \Omega((n/\delta)^2)$, and thus the $\delta$-shallowness assumption does not yield
an improvement over the general case.

\section{Overview of Haussler's Approach}
\label{app:Haussler_proof}

Let $\V \subseteq \{0,1\}^n$ be a collection of indicator vectors of primal shatter dimension $d$.
We denote its VC-dimension by $d_0$; as discussed in Section~\ref{sec:intro} $d_0 = O(d\log{d})$.



We first form a probability distribution $P$ over $\V$, implying that $\V$ can be viewed\footnote{For the time being, $P$ is an arbitrary distribution, but later on (Lemma~\ref{lem:expected_var1}) it is taken to be the uniform distribution in the obvious way, where each vector in $\V$ is equally likely to be chosen. This distribution, however, may not remain uniform after the projection of $\V$ onto a proper subsequence $I' = (i_1, \ldots, i_m)$ of $m < n$ indices, as several vectors in $\V$ may be projected onto the same vector in $\V_{|_{I'}}$.} as an 
$n$-dimensional random variable taking values in $\{0,1\}^n$. Thus its components $\V_i$, $i=1, \ldots, n$,
represent $n$ correlated indicator random variables (Bernulli random variables), and each of their values is determined  
by randomly selecting a vector $\vv \in \V$, and letting $\V_i$ be the $i$th component of $\vv$.
The variance of a Bernulli random variable $B$ is known to be $\Prob[B=1] \Prob[B=0]$, and then, for a sequence
$B_1, \ldots, B_m$ of Bernulli random variables, the \emph{conditional variance} of $B_m$ given $B_1, \ldots, B_{m-1}$
is defined as
$$
\Var(B_m | B_1, \ldots, B_{m-1}) = \sum_{\vv \in \{0,1\}^{m-1}} \Prob(\vv) \Prob(B_m = 1 | \vv) \left(1 - \Prob(B_m =1 | \vv)\right) ,
$$
where $\Prob(\vv) = \Prob(B_1 = \vv_1, B_2 = \vv_2, \ldots, B_{m-1} = \vv_{m-1})$, and 
$\Prob(B_m = 1 | \vv) = \Prob(B_m = 1 | B_1 = \vv_1, B_2 = \vv_2, \ldots, B_{m-1} = \vv_{m-1})$.    

A key property in the analysis of Haussler~\cite{Haussler-95} lies in the density of a \emph{unit distance graph} $G =(\V,E)$, 
defined over $\V$, whose edges correspond to all pairs $\uu, \vv \in \V$, whose symmetric difference distance is (precisely) $1$. 
In other words, $\uu$, $\vv$ appear as neighbors on the unit cube $\{0,1\}^n$.  
It has been shown by Haussler~\etal~\cite{HLW-90} that the density of $G$ is bounded by the VC-dimension of $\V$,
that is, $|E|/|\V| \le d_0$; see also~\cite{Haussler-95} for an alternative proof using the technique of ``shifting''.\footnote{We cannot guarantee
  such a relation when the VC-dimension $d_0$ is replaced by the primal shatter dimension $d$, and therefore we proceed with the analysis using this ratio.} 
Then, this low density property is exploited in order to show that once we have chosen $(n-1)$ coordinates of the random
variable $\V$, the variance in the choice of the remaining coordinate is relatively small. 
That is:

\begin{lemma}[~\cite{Haussler-95}]
  \label{lem:sum_variance}
  For any distribution $P$ on $\V$,
  $$
  \sum_{i=1}^n \Var(\V_i | \V_1, \ldots, \V_{i-1}, \V_{i+1}, \ldots, \V_n) \le d_0 .
  $$
\end{lemma}



As observed in \cite{Haussler-95}, Lemma~\ref{lem:sum_variance} continues to hold on any restriction of $\V$ to a sequence 
$I' = \{i_i,\ldots, i_m\}$ of $m \le n$ indices. Indeed, when projecting $\V$ onto $I'$ the VC-dimension in the resulting set system
remains $d_0$. Furthermore, the conditional variance is now defined w.r.t. the induced probability distribution on $\V_{|_{I'}}$ in the 
obvious way, where the probability to obtain a sequence of $m$ values corresponds to an appropriate marginal distribution, 
that is, $\Prob_{|_{I'}}(u_1, \ldots, u_m) = \Prob(\vv \in \V \mid v_{i_j} = u_j, 1 \le j \le m)$.
With this observation, we can rewrite the inequality stated in Lemma~\ref{lem:sum_variance} as
$$
\sum_{i=1}^m \Var(\V_{i_{j}} | \V_{i_{1}}, \ldots, \V_{i_{j-1}}, \V_{i_{j+1}}, \ldots, \V_{i_{m}}) \le d_0 .
$$
If $I'$ is a sequence chosen uniformly at random (over all such $m$-tuples), then
when averaging over all choices of $I'$ 
we clearly obtain: 
$$
\EE\left[ \sum_{i=1}^m \Var(\V_{i_{j}} | \V_{i_{1}}, \ldots, \V_{i_{j-1}}, \V_{i_{j+1}}, \ldots, \V_{i_{m}})\right] \le d_0 ,
$$
or
$$
\sum_{i=1}^m \EE\left[\Var(\V_{i_{j}} | \V_{i_{1}}, \ldots, \V_{i_{j-1}}, \V_{i_{j+1}}, \ldots, \V_{i_{m}})\right] \le d_0 ,
$$
by linearity of expectation.
In fact, by symmetry of the random variables $\V_{i_{j}}$ (recall that $I'$ is a random $m$-tuple) each of the summands in the above inequality 
has an equal contribution, and thus, in particular (recall once again that the expectation is taken over all choices of $I'= \{i_i,\ldots, i_m\}$):
\begin{equation}
  \label{eq:cond_var}
  \EE_{I'} \left[\Var(\V_{i_m} | \V_{i_{1}},\ldots, \V_{i_{m-1}})\right] \le \frac{d_0}{m} ,
\end{equation}
where we write $\EE_{I'}[\cdot]$ to emphasize the fact that the expectation is taken over all choices of $I'$.
The above bound is now integrated with the next key property: 

\begin{lemma}[~\cite{Haussler-95}]
  \label{lem:expected_var1}
  Let $\V$ be $\delta$-separated subset of $\{0,1\}^n$, for some $1 \le \delta \le n$ integer, and form a uniform distribution $P$
  on $\V$. Let $I = (i_1, \ldots, i_{m-1})$ be a sequence of $m-1$ distinct indices between $1$ and $n$, where $m$ is any integer between $1$ 
  and $n$. Suppose now that another index $i_m$ is drawn uniformly at random from the remaining $n-m+1$ indices. Then
  $$
  \EE \left[ \Var\left(\V_{i_m} | \V_{i_1} \ldots, \V_{i_{m-1}} \right) \right] \ge
  \frac{\delta}{2(n-m+1)} \left(1 - \frac{ |\V_{|_{I}}| }{|\V|} \right) ,
  $$
  where the conditional variance is taken w.r.t. the distribution $P$, 
  and the expectation is taken w.r.t. the random choice of $i_m$.
\end{lemma}

We now observe that when the entire sequence $I' = (i_1, \ldots, i_m)$ is chosen uniformly at random, 
then the bound in Lemma~\ref{lem:expected_var1} continues to hold when averaging on the entire sequence $I'$
(rather than just on $i_m$), that is, we have: 
\begin{equation}
  \label{eq:expected_var2}
  \EE_{I'} \left[ \Var\left(\V_{i_m} | \V_{i_1} \ldots, \V_{i_{m-1}} \right) \right] \ge
  \EE_{I'} \left[\frac{\delta}{2(n-m+1)} \left(1 - \frac{ |\V_{|_{I}}| }{|\V|} \right)\right] 
\end{equation}
$$
  =	
  \frac{\delta}{2(n-m+1)} \left(1 - \frac{ \EE_{I} [ |\V_{|_{I}}| ] }{|\V|} \right) .
$$
Note that under this formulation, $|\V_{|_{I}}|$ (the number of sets in the projection of $\V$ onto $I$) is a random variable 
that depends on the choice of $I = (i_1, \ldots, i_{m-1})$, in particular, since it does not depend on the choice of $i_m$, we have
$\EE_{I'} [ |\V_{|_{I}}| ] = \EE_{I} [ |\V_{|_{I}}| ]$.

The analysis of Haussler~\cite{Haussler-95} then proceeds as follows.
We assume that $\V$ is $\delta$-separated as in Lemma~\ref{lem:expected_var1} (then a bound on $|\V|$ is the actual bound
on the packing number), and then choose 
$$
m := \left\lceil{\frac{(2d_0 + 2)(n+1)}{\delta + 2d_0 + 2} }\right\rceil ,
$$
indices $i_1, \ldots, i_m$ uniformly at random without replacements from $[n]$
(without loss of generality, we can assume that $\delta \ge 3$ as, otherwise, we set the bound on the packing
to be $O(n^{d_1} k^{d-d_1})$, as asserted by the $(d,d_1)$ Clarkson-Shor property of $\V$.
Moreover, we can assume, without loss of generality, $n \ge d_0, \delta$, and thus we have $m \le n$).
Put $I' = \{i_1, \ldots, i_m\}$, $I = I' \setminus \{i_m\}$.
Then the analysis in~\cite{Haussler-95} combines the two bounds in Inequalities~(\ref{eq:cond_var}) and~(\ref{eq:expected_var2})
in order to derive an upper bound on $|\V|$, from which the bound in the Packing Lemma (Theorem~\ref{thm:packing}) is obtained.
Specifically, using simple algebraic manipulations, we obtain:
\begin{equation}
  \label{eq:bound_V_by_expectation}
  |\V| \le \frac{\EE_{I}\left[ |\V_{|_{I}}| \right]}{1 - \frac{2d_0(n-m+1)}{m\delta} }
\end{equation}
It has been shown in~\cite{Haussler-95} that due to our choice of $m$, we have:
\begin{equation}
  \label{eq:ratio}
\frac{2d_0(n-m+1)}{m\delta} \le \frac{d_0}{d_0 + 1} ,
\end{equation}
from which we obtain Inequality~(\ref{eq:packing_bound}), as asserted.
This completes the description of our first extension of Haussler's analysis, as described in Section~\ref{sec:prelim}.

In order to apply our second extension, we observe that one only needs to assume $m \le n$ when obtaining 
Inequalities~(\ref{eq:cond_var}) and~(\ref{eq:expected_var2}). In addition, the choice of $m$ in Inequality~(\ref{eq:bound_V_by_expectation})
can be made slightly larger, since the term $\frac{2d_0(n-m+1)}{m\delta}$ in Inequality~(\ref{eq:ratio}) is a decreasing function of $m$, 
as can easily be verified.
Recall that in our analysis we replace $m$ by $m_j := m \log^{(j)}{(n/\delta)}$, where $2 \le j \le \log^{*}{(n/\delta)}$,
in which case we still obtain $|\V| \le (d_0 + 1)\EE_{I_j}\left[ |\V_{|_{I_j}}| \right]$, where $I_j$ is a set of $m_j - 1$ indices chosen
uniformly at random without replacements from $[n]$.


\section{The First Iteration}
\label{app:compact_rep}

Our proof relies on the notions of ``relative approximations'' and ``$\eps$-nets'', define below:

\paragraph*{Relative approximations and $\eps$-nets.}

We mentioned in the introduction the notion of \emph{relative $(\eps, \eta)$-approximations}.
We now define them formally: Following the definition from~\cite{HS-11},
given a set system  $\V \subseteq \{0,1\}^n$ and two parameters, $0 < \eps <1$ and $0 < \eta < 1$,
we say that a subsequence $I$ of indices is a \emph{relative $(\eps, \eta)$-approximation} if it satisfies,
for each vector $\vv \in \V$,
$$
\left | \frac{\|\vv_{|_I}\|}{|I|} -  \frac{\|\vv\|}{n} \right| \le \eta\frac{\|\vv\|}{n}  ,
\quad \mbox{if $ \frac{\|\vv\|}{n}\ge \eps$,} \quad \mbox{and}
$$
$$
\left | \frac{\|\vv_{|_I}\|}{|I|} -  \frac{\|\vv\|}{n} \right| \le \eta\eps , \quad \mbox{otherwise,}
$$
where $\vv_{|_I}$ is the projection of $\vv \in \V$ onto $I$.

As observed by Har-Peled and Sharir~\cite{HS-11}, the analysis of Li~\etal~\cite{LLS-01} implies that if $\V$ has primal shatter dimension $d$,
then a random sample of $\frac{c d\log{(1/\eps)}}{\eps \eta^2}$ indices (each of which is drawn independently) is a relative 
$(\eps,\eta)$-approximation for $\V$  with constant probability, where $c > 0$ is an absolute constant.
More specifically, success with probability at least $1-q$ is guaranteed if one samples
$\frac{c (d\log{(1/\eps)} + \log{(1/q)})}{\eps \eta^2}$ indices.\footnote{We note that although in the original analysis for this bound 
  $d$ is the VC-dimension, this assumption can be replaced by having just a primal shatter dimension $d$; 
  see, e.g.,~\cite{Har-Peled-11} for the details of the analysis.}

It was also observed in~\cite{HS-11} that \emph{$\eps$-nets} 
arise as a special case of relative $(\eps,\eta)$-approximations.
Specifically, an $\eps$-net is a subsequence of indices $I$ with the property that any vector $\vv \in \V$ with $\|\vv\|\ge n \eps$
satisfies $\|\vv_{|_I}\| \ge 1$.
In other words, $N$ is a hitting set for all the ``long'' vectors.
In this case, if we set $\eta$ to be some constant fraction, say, $1/4$, then a relative $(\eps,1/4)$-approximation becomes
an $\eps$-net. Moreover, a random sample of $O\left(\frac{d\log{(1/\eps)} + \log{(1/q)}}{\eps}\right)$ indices (with an appropriate choice 
of the constant of proportionality) is an $\eps$-net for $\V$, with probability at least $1-q$; see~\cite{HS-11} for further details.

\paragraph*{Proof of Lemma~\ref{lem:compact_rep}.}
  In order to show (i), we first form the set system corresponding to all symmetric difference pairs induced by $\V$.
  That is, we form the vector set $\D$, where $\D = \{ (|\uu_1 - \vv_1|, \ldots, |\uu_n - \vv_n|) \mid \uu, \vv \in \V \}$.
  Since we assume that $\V$ is $\delta$-separated, we have $\|\ww\| \ge \delta$, for each $\ww \in \D$.

  We now construct an 
  \emph{$\eps$-net} for $\D$, with $\eps = \delta/n$. By our discussion above 
  a sample $I_1$ of $O(d (n/\delta) \log{(n/\delta)})$ indices has this
  property with probability greater than, say, $3/4$ (for a sufficiently large constant of proportionality).
  %
  Thus, by definition, any vector $\ww \in \D$ (recall that its length is at least $\delta$) 
  must satisfy $|\ww_{|_{I_1}}| \ge 1$, where $\ww_{|_{I_1}}$ denotes the projection of $\ww$ onto $I_1$. 
  %
  %
  But this implies that we must have $\uu_{|_{I_1}} \neq \vv_{|_{I_1}}$, for each pair $\uu, \vv \in \V$,
  and thus $\uu$, $\vv$ must be mapped to distinct vectors in the projection of $\V$ onto $I_1$,
  from which property (i) follows.

  In order to have property (ii) we observe that the same sample $I_1$ is also a \emph{relative $(\delta/n, 1/4)$-approximation}
  for $\V$ with probability at least $3/4$ (for an appropriate choice of the constant of proportionality). 
  Given this property of $I_1$, this implies that any vector $\vv \in \V$ satisfies
  $$
  \left| \frac{\|\vv_{|_{I_1}}\|}{m_1} - \frac{\|\vv\|}{n} \right| \le \frac{1}{4} \cdot \frac{\|\vv\|}{n} ,
  $$
  if $\|\vv\|/n \ge \delta/n$, and
  $$
  \left| \frac{\|\vv_{|_{I_1}}\|}{m_1} - \frac{\|\vv\|}{n} \right| \le \frac{1}{4} \cdot \frac{\delta}{n} ,
  $$
  otherwise.
  Since $\|\vv\| \le k$, and $k \ge \delta/2$ by assumption, it is easy to verify that we always have
  $\|\vv_{|_{I_1}}\| \le 3/2 \cdot k m_1/n$. In other words, $\|\vv_{|_{I_1}}\| = O(k\cdot m_1/n)$, as asserted.

  Combining the two roles of $I_1$ (each with probability $3/4$), 
  it follows that it is both a $(\delta/n)$-net for $\D$ and a relative $(\delta/n, 1/4)$-approximation for $\V$, with probability at least $1/2$,
  and thus it satisfies properties (i)--(ii) with this probability.
  This completes the proof of the lemma.
$\Box$

\section{Proof of Lemma~\ref{lem:exp_decay}}
\label{app:exp_decay}

\begin{proof}


  We first observe that the length of $\vv_{|_{I_j}}$ is a random variable with a hypergeometric distribution.
  Indeed, this is precisely the question of uniformly choosing $m_j-1$ elements at random (into our set $I_j$) from a given set 
  of $n$ elements without replacements, and then, for a given $\| \vv \|$-element subset of the full set (recall that the length of
  $\vv$ corresponds to the cardinality of an appropriate subset in the set system), we consider how many of its 
  elements have been chosen into $I_j$. Specifically, we have:
  $$
  \Prob\left[\| \vv_{|_{I_j}} \| = s \right] = \frac{{\| \vv\| \choose s} {{n -\| \vv\|} \choose {m_j - 1 - s} }}{ {n \choose {m_j - 1}} } ,
  $$
  for each non-negative integer $s \le \min\{\| \vv \|, m_j-1\}$.
  
  Our goal is to show a Chernoff-type bound over the probability that $\| \vv_{|_{I_j}} \|$ deviates from its expectation.
  However, we face the difficulty that the corresponding indicator variables are not independent, and thus we cannot apply a
  Chernoff bound directly (see, e.g.~\cite{AS-00}). Nevertheless, in our scenario a Chernoff bound is still applicable, 
  this can be viewed by various approaches, see, e.g.,~\cite{Doerr-11, Mulzer-note, PS-97}. 
  For the sake of completeness we describe the proof in detail, and rely on the analysis of Panconesi and Srinivasan~\cite{PS-97}, 
  which implies that when the underlying indicator variables are ``negatively correlated'', one can still apply a Chernoff bound 
  (see also~\cite{Doerr-11}).
  
  We enumerate all non-zero coordinates of $\vv$ in an arbitrary order, 
  let $L = \{l_1, \ldots l_{\|\vv\|} \}$ be this set of indices (in this notation we ignore all the zero-coordinates), 
  and attach an indicator variable $X_i$ to each index $l_i \in L$, which is defined to be one if and only if $l_i \in I_j$
  (in other words, the corresponding element in the underlying set induced by $\vv$ has been chosen to be included into the sample 
  of the $m_j-1$ elements). According to this notation, $\| \vv_{|_{I_j}} \|$ is represented by the sum $X = \sum_{i=1}^{\| \vv\|} X_{i}$.
  It is now easy to verify that $\Prob[X_{i} = 1] = (m_j - 1)/n$, and by linearity of expectation 
  $\EE [\| \vv_{|_{I_j}} \|] = \| \vv \| \cdot (m_j - 1)/n$.
  
  However, the variables $X_{i}$ are not independent due to our probabilistic model (that is, $I_j$ is chosen without replacements),
  albeit, they are \emph{negatively correlated}. This implies that for each subset $K \subseteq \{1, \ldots, \| \vv \| \}$
  $$
  \Prob[\bigwedge_{i \in K} X_i = 0] \le \prod_{i \in K} \Prob[X_i = 0] , 
  $$
  and
  $$
  \Prob[\bigwedge_{i \in K} X_i = 1] \le \prod_{i \in K} \Prob[X_i = 1] .  
  $$
  Indeed, following the considerations in~\cite{Doerr-11}, let us show first the latter inequality.
  Put $L_{K} = \bigcup_{i \in K}\{l_i\}$.
  Then $\Prob[\bigwedge_{i \in K} X_i = 1] = \Prob[L_{K} \subseteq I_j]$, and since $I_j$ is uniformly chosen, in order to bound the latter 
  we need to take the proportion between the number of subsets of size $m_j-1$ that contain $L_{K}$ and the entire number of subsets of size
  $m_j-1$ that can be chosen from an $n$-element set. 
  Hence
  $$
  \Prob[L_{K} \subseteq I_j] = \frac{{{n - |K|} \choose {m_j - 1 - |K|}}} {{n \choose {m_j-1}}} = 
  \frac{(m_j - 1)(m_j - 2) \cdots (m_j - |K|)} {n(n-1) \cdots (n-|K|+1)} ,
  $$
  and the latter is smaller than $\left(\frac{m_j - 1}{n}\right)^{|K|}$, as is easily verified.
  Using similar arguments for the first correlation inequality, we obtain 
  $$
  \Prob[\bigwedge_{i \in K} X_i = 0] =  \frac{{{n - |K|} \choose {m_j - 1}}} {{n \choose {m_j - 1}}} < \left(1 - \frac{m_j - 1}{n} \right)^{|K|} .
  $$

  We are now ready to apply~\cite[Theorem 3.4]{PS-97} stating that if the indicator variables $X_i$ are negatively correlated then
  (recall that $X = \sum_{i=1}^{\| \vv\|} X_{i}$)\footnote{We note that in the original formulation in~\cite{PS-97}, one needs to have a set of 
    \emph{independent} random variables $\hat{X}_{i}$, $i \in \{1, \ldots, \| \vv \| \}$ with $\hat{X} = \sum_{i=1}^{\| \vv\|} \hat{X}_{i}$, 
    such that $\EE[X] \le \EE[\hat{X}]$. In the scenario of our problem $\hat{X}_i$ is taken to be a Bernulli indicator random variable, 
    which takes value one with probability $(m_j - 1)/n$, in which case $\EE[X] = \EE[\hat{X} ] = \| \vv \| \cdot (m_j - 1)/n$. }:
  $$
  \Prob[X > \rho \EE[X] ] < \left(\frac{e^{\rho-1}} {\rho^{\rho}}\right)^{\EE[X]} ,
  $$
  for any $\rho > 1$. In particular, when $\rho \ge 2e$ (where $e$ is the base of the natural logarithm), the latter term
  is bounded by $2^{-\rho \EE[X] }$.
  Recall that we assumed $\| \vv \| \le k$, and thus $\EE[X] = \EE[\| \vv_{|_{I_j}} \|] \le  k \cdot (m_j - 1)/n$. 
  Thus, for any $t \ge 2e$, we obtain:
  $$
  \Prob\left[\| \vv_{|_{I_j}} \| > t k(m_j - 1)/n\right]  = 
  \Prob\left[\| \vv_{|_{I_j}} \| > \frac{t k(m_j - 1)/n}{\EE[\| \vv_{|_{I_j}} \|]} \cdot \EE[\| \vv_{|_{I_j}} \|] \right] ,
  $$
  observe that in this case $\rho := \frac{t k(m_j-1)/n}{\EE[\| \vv_{|_{I_j}} \|]} \ge 2e$, due to our assumption on $t$ and the fact that
  $\frac{k(m_j-1)/n}{\EE[\| \vv_{|_{I_j}} \|]} \ge 1$, and thus 
  the latter term is bounded by:
  $$
  2^{- \frac{t k(m_j - 1)/n}{\EE[\| \vv_{|_{I_j}} \|]}  \EE[\| \vv_{|_{I_j}} \|] } =  2^{-t k(m_j - 1)/n} ,
  $$
  as asserted.
\end{proof}

\section{Approximating the Minimum Spanning Tree}
\label{app:approx_mst}

Given a set system $(X,\Sigma)$ as in Section~\ref{sec:spanning_trees}, our goal is to approximate the spanning tree
with minimum conflicts, using the metric embedding approach of Har-Peled and Indyk~\cite{HI-02}. 
The original settings studied in~\cite{HI-02} are dual set systems of halfspaces and points (as in~\cite{Welzl-92}), 
however, most of the steps in their analysis can be applied in our case as well, and therefore we mainly emphasize the modifications 
required by our analysis.

Specifically, we proceed as follows. For any subset $P \subseteq X$, we form a mapping $f_P : \Sigma \rightarrow {\integers}$,
which maps a region (i.e., a set) in $\Sigma$ to a unique integer ID in the projection $\Sigma_{|_P}$ (see~\cite{HI-02} for the existence 
of such a mapping). Next, we embed the symmetric difference distance between each pair of sets $S, S' \in \Sigma$ to the Hamming space
${\integers}^\mu$, where $\mu = \polylog \{m\}$, using a collection $\P$ of $\mu$ random subsets $P_1, \ldots, P_\mu$ of $X$. 
Here, the Hamming space consists of all ID vectors $f_{\P}(S) = (f_{P_1}(S), \ldots, f_{P_{\mu}}(S))$, for each $S \in \Sigma$.
Then, a main property of the analysis in~\cite{HI-02} is to have a low distortion between the symmetric difference distance 
of pairs of sets $S$, $S'$ and the Hamming distance between their two corresponding vectors $f_{\P}(S)$, $f_{\P}(S')$
(which is the number of coordinates where they disagree).\footnote{We note that this part of the analysis in~\cite{HI-02} is sufficiently general to include our setting as well.}

Once the mapping $f_{\P}(\cdot)$ is computed, for all sets $S \in \Sigma$, 
one can use existing machinery to support dynamic approximate nearest-neighbor in order to compute $(1+\eta)$-approximation for the 
spanning tree in time $O(m^{1 + 1/(1+\eta)})$ (this is described in~\cite{HI-02} is detail, see also~\cite{Epp-95, IM-98, KOR-00}). 

Thus a main task is to compute $f_{\P}(\cdot)$ efficiently. Whereas the computation of $f_{\P}$ in~\cite{HI-02}
is based on the construction of many faces in an arrangement of lines (for planar settings) and an ``intersection-searching'' machinery
(for settings in higher dimensions), we need to resort to a \emph{range-counting} machinery (see, e.g.,~\cite{AE-97}). In this case one
can obtain a \emph{canonical representation} for the subset of points stored in each $S \in \Sigma$, and then produce the IDs based on the
canonical sets (rather than the explicit sets). Skipping the technical (and straightforward) details, we conclude that the running time for 
computing $f_{\P}(\cdot)$ is dominated, up to a polylogarithmic factor, by the task of computing, for each $S \in \Sigma$, 
the number of points that it contains from $P_1, \ldots, P_\mu$. By existing range searching machinery (see, e.g.,~\cite{AE-97}), 
it follows that this time is (i) $O^{*}(n^{d/(d+1)} m^{d/(d+1)} + n + m)$, if $\Sigma$ is a set of halfspaces or parallel slabs in $d$-space, 
and (ii) $O^{*}(n^{(d+1)/(d+2)} m^{(d+1)/(d+2)} + n + m)$, if $\Sigma$ is a set of balls in $d$-space. 
Thus from the above considerations we obtain the bounds stated in Corollary~\ref{cor:approx_mst}.

\end{document}

%% file: rectangles.pdf_t
\begin{picture}(0,0)%
\includegraphics{rectangles.pdf}%
\end{picture}%
\setlength{\unitlength}{1973sp}%
\begingroup\makeatletter\ifx\SetFigFont\undefined%
\gdef\SetFigFont#1#2#3#4#5{%
  \reset@font\fontsize{#1}{#2pt}%
  \fontfamily{#3}\fontseries{#4}\fontshape{#5}%
  \selectfont}%
\fi\endgroup%
\begin{picture}(3334,3334)(5989,-4433)
\end{picture}%